\documentclass[11pt]{article}

\usepackage{fullpage}

\usepackage{amsmath,amsthm}

\usepackage{microtype}

\usepackage{epic,ecltree}    

\usepackage{tikz}

\usepackage[ruled,vlined,nofillcomment]{algorithm2e}
\usepackage{multirow}
\newcommand{\kreis}[1]{\fcolorbox{white}{black!15}{$#1$\vphantom{1}}$\ \Longrightarrow$}
\newcommand{\rkreis}[1]{$\Longleftarrow$\ \fcolorbox{white}{black!15}{\vphantom{1}$#1$}}

\newcommand{\clset}{\psi}
\newcommand{\theConstdrei}{{\ensuremath 1.51426}}
\newcommand{\theConstvier}{{\ensuremath 1.60816}}
\newcommand{\struct}{\textsc{struct}}
\newcommand{\doi}[1]{{\texttt{\small #1}}}

\newtheorem{theorem}{Theorem}
\newtheorem{lem}{Lemma}

\begin{document}

\title{\LARGE\bf
Exploiting Independent Subformulas: A
Faster Approximation Scheme for \#\textit{k}-SAT}

\author{Manuel Schmitt \qquad Rolf Wanka\\[2mm]
Department of Computer Science\\
University of Erlangen-Nuremberg, Germany\\
{\{\tt manuel.schmitt, rolf.wanka\}@cs.fau.de}
}

\date{ }

\maketitle

\begin{abstract}
We present an improvement on Thurley's recent randomized approximation
scheme for $\#k$-SAT where the task is to count the number of
satisfying truth assignments of a Boolean function $\Phi$ given
as an $n$-variable $k$-CNF.
We introduce a novel way to identify independent substructures
of $\Phi$
and can therefore reduce
the size of the search space considerably.
Our randomized algorithm works for any $k$.
For $\#3$-SAT, it runs in
time $O(\varepsilon^{-2}\cdot\theConstdrei^n)$,
for $\#4$-SAT, it runs in time $O(\varepsilon^{-2}\cdot\theConstvier^n)$,
with error bound $\varepsilon$.
\end{abstract}



\section{Introduction}
\label{sec:intro}

\paragraph{Background.} The satisfiability problem (SAT) is one of the
classical and central problems in algorithm theory.
Its prominent role in Computer Science has even been compared~\cite{S:03}
to the one that Drosophila (the fruit fly) has in Genetics.
Given a Boolean formula~$\Phi$ in conjunctive normal
form (CNF) on $n$ variables with $m$ clauses,
it has to be determined whether there is a satisfying
assignment for~$\Phi$ (and in this case, to determine one) or not.
If every clause of $\Phi$ has length at most $k$, $\Phi$ is called
a $k$-CNF and the problem is dubbed $k$-SAT.
It is well known (for a comprehensive overview, see~\cite{GJ:79})
that $k$-SAT is NP-complete
for any $k\ge3$, and that it can be solved in time linear in the
input length for $k=2$~\cite{APT:79}.
So it is generally assumed that there is
no polynomial time algorithm solving $k$-SAT for $k\ge 3$.
In particular, $3$-SAT has attracted much attention
because of its ``borderline'' status.

There is a rich history of developing both deterministic and randomized
algorithms with running time $o(2^n)$ solving $k$-SAT.
The currently fastest deterministic algorithm
for $3$-SAT runs in time\footnote{In this context, the notion $O^*(.)$ is commonly used to suppress
factors that are of size $2^{o(n)}$.}
$O^*(1.3303^n)$~\cite{MTY:11},
the fastest randomized algorithm has
a running time of $O^*(\log(\delta^{-1})\cdot 1.30704^n)$~\cite{H:11}.
In the randomized setting, the use of $\delta$ means the following:
If $\Phi$ is not satisfiable, the algorithm returns the correct answer.
If $\Phi$ is satisfiable, it returns with probability
$1-\delta$ a satisfying assignment.
Table~\ref{tab:KnownRes} presents all best running times
currently known to solve $k$-SAT.

For many combinatorial problems including $k$-SAT, it is often not only
important to determine one solution (if it exists), but also
to determine 
the number of all different solutions.
A famous example from statistical physics is the computation of the
number of configurations in monomer-dimer systems 
(for an overview, see~\cite{JS:96}).
The complexity class that corresponds to these \emph{counting problems}
is $\#$P, and $\#$SAT, the problem to determine the number
of satisfying assignments, is well-known to be $\#$P-complete.
More exactly, let $\#k$-SAT denote the problem to determine $\#\Phi$,
i.\,e., for input $\Phi$ being a $k$-CNF,
the number of satisfying assignments.
Then, it is known~\cite{V:79} that
$\#k$-SAT is \#P-complete for $k\ge 2$.

\begin{table}\small
\centering
\caption{\label{tab:KnownRes}Previous and new results for $k$-SAT and $\#k$-SAT where the input $k$-CNF
has $n$ variables and $m$ clauses. The times are given in $O^*(.)$ notation.
$\beta_k$ is the base-$2$ logarithm of the base of the running time in column ``$k$-SAT rand''.
For definition of $\mu_k$, see Sec.~\ref{subsec:Thur};
$\psi_k$ is the largest root of $1-2z^k+z^{k+1}=0$;
$2^{1/(2-\beta_k)}>\alpha_k$.}

\medskip

\begin{tabular}{@{}l|lll|lll@{}}
 & \multicolumn{1}{c}{$k$-SAT} & \multicolumn{1}{c}{$k$-SAT} &  & \multicolumn{1}{c}{$\#k$-SAT} & \multicolumn{1}{c}{$\#k$-SAT}& \multicolumn{1}{c}{$\#k$-SAT} \\
 & \multicolumn{1}{c}{deterministic} & \multicolumn{1}{c}{rand} & \multicolumn{1}{c|}{$\beta_k$} & \multicolumn{1}{c}{exact} & \multicolumn{1}{c}{rand, prev.}& \multicolumn{1}{c}{this paper}\\ \hline
 $k=2$ & $n+m$\hfill\cite{APT:79} & \multicolumn{1}{c}{---} & \multicolumn{1}{c|}{---} & $1.2377^n$\hfill\cite{W:08} & \multicolumn{1}{c}{---}& \multicolumn{1}{c}{---} \\
 $k=3$ & $1.3303^n$\hfill\cite{MTY:11} & $1.30704^n$\hfill\cite{H:11} & $0.3864$ & $1.6423^n$\hfill\cite{K:07} & $1.5366^n$\hfill\cite{T:12} & $\theConstdrei^n$ \\
 $k=4$ & $1.5^n$\hfill\cite{MS:11} & $1.46899^n$\hfill\cite{H:11} & $0.5548$  & $1.9275^n$\hfill\cite{D:91,Z:96} & $1.6155^n$\hfill\cite{T:12} & $\theConstvier^n$ \\
 $k\ge 5$ & $\big(\frac{2\cdot(k-1)}{k}\big)^n$~\cite{MS:11} & $2^{(1-\mu_k/(k-1))\cdot n}$~\cite{PPSZ:05} & $1-\frac{\mu_k}{k-1}$ & $\psi_k^n$\hfill\cite{Z:96} & $2^{1/(2-\beta_k)\cdot n}$~\cite{T:12} & $\alpha_k^n$ (Sec.~\ref{subsec:HierIsses})
\end{tabular}

\end{table}

\paragraph{Topic of this work.}
In the area of combinatorial counting problems, there is also the
problem of approximating the wanted number.
In particular, there is the task to develop so-called \emph{randomized approximation schemes}
that receive as input $\Phi$ and
an arbitrarily small bound $\varepsilon$ on the
maximum admissible error and that compute with some fixed probability
greater $1/2$ an $\varepsilon$-estimate of $\#\Phi$
(for exact definitions, see Sec.~\ref{sec:prelim}).
In a recent paper, Thurley~\cite{T:12} presents
such a randomized approximation scheme for $\#k$-SAT
that has, for $k=3$,
running time $O^*(\varepsilon^{-2}\cdot 1.5366^n)$,
and for $k=4$, $O^*(\varepsilon^{-2}\cdot 1.6155^n)$.
A detailed description of Thurley's algorithm is presented
in Sec.~\ref{sec:prelim}.
Table~\ref{tab:KnownRes} also presents all best running times
currently known to solve $\#k$-SAT.

A different approach by Impagliazzo et al.~\cite{IMP:12}
leads to a randomized Las Vegas algorithm for $\#k$-SAT 
that always returns the exact solution and has expected
running time $O^*(2^{(1-1/(30k))n})$.
Note that for any $k$, Thurley's algorithm is faster than
this method.

\paragraph{New Results.}
We present a randomized approximation scheme
for $\#k$-SAT that takes the input $k$-CNF much more into
account than Thurley's algorithm.
In particular, we present a method that determines a large set
of maximal independent subformulas of $\Phi$.
I.\,e., the subformulas have no variables in common and
can therefore be treated independently.
As they are maximal, they convert the remaining clauses into clauses
of length $k-1$.
Hence, the search space is substantially reduced.
Our scheme,  which works for any $\#k$-SAT instance,
has for $\#3$-SAT running time $O(\varepsilon^{-2}\cdot\theConstdrei^n)$, and
for $\#4$-SAT, it works in time $O(\varepsilon^{-2}\cdot\theConstvier^n)$.
Note that our scheme is for all $k$ faster than Thurley's scheme.

\paragraph{Organization of Paper.}
In the next section, we define the necessary terms, and we give a
comprehensive description of Thurley's randomized approximation scheme.
In Sec.~\ref{sec:FirstImpr}, we present a first improvement that
exploits single clauses.
Generalizing this approach and building upon each other, we present further improvements
based on large sets of maximal independent clauses (Sec.~\ref{sec:NewRas}),
and on large sets of maximal independent subformulas (Sec.~\ref{sec:FinalRAS}).


\section{Elimination Trees, Monte Carlo Counting, and Thurley's Algorithm}
\label{sec:prelim}

Let $\Phi$ be a $k$-CNF, i.\,e., a Boolean function given in conjunctive normal
form with $n$ different variables $x_1,\ldots,x_n$ on $m$ different
clauses such that every clause has length at most $k$.
For an arbitrary Boolean formula $\phi$, let ${\rm Var}(\phi)$
denote the variables that occur in $\phi$.
Let $b:{\rm Var}(\phi)\to\{0,1\}$ be a \emph{partial}
assignment of truth values to the variables in $\phi$.
By $\phi_b$ we denote the formula we obtain from $\phi$ by
fixing in $\phi$ the variables according to $b$.

\newcommand{\hh}{\rule[-4pt]{0pt}{16pt}}
\begin{figure}
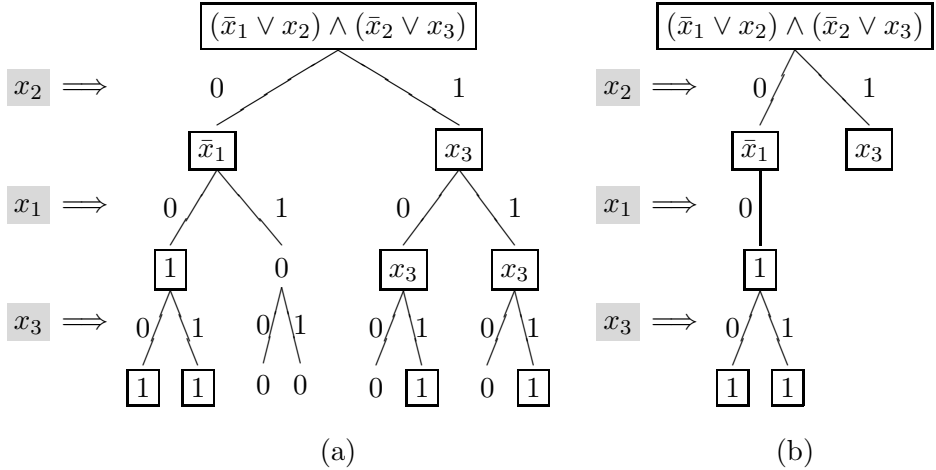

\centering
\scalebox{1}{
\begin{tabular}{@{}cc@{\qquad}cc@{}}
\mbox{%
\setlength{\GapDepth}{1cm}%
\setlength{\GapWidth}{4mm}%
\setlength{\EdgeLabelSep}{0.5cm}%
\drawwith{\dottedline[ ]{3}}%
\begin{bundle}{\hh}
  \chunk[\protect\kreis{x_2}\protect\vphantom{1}]{
    \begin{bundle}{\hh}
      \chunk[\protect\kreis{x_1}]{
          \begin{bundle}{\hh}
             \chunk[\protect\kreis{x_3}]{\hh}
          \end{bundle}
      }
    \end{bundle}
  }
\end{bundle}%
}
&
\mbox{\setlength{\GapDepth}{1cm}%
\setlength{\GapWidth}{3mm}%
\setlength{\EdgeLabelSep}{0.5cm}%
\begin{bundle}{\fbox{$(\bar x_1\lor x_2)\land (\bar x_2\lor x_3)$}\hh}
   \chunk[0]{
      \begin{bundle}{\fbox{$\bar x_1$} \hh}
         \chunk[0]{
            \begin{bundle}{\fbox{$1\vphantom{x_2}$}\hh}
              \chunk[0]{\fbox{$1$}\hh}
              \chunk[1]{\fbox{$1$}\hh}
             \end{bundle}
          }
         \chunk[1]{
            \begin{bundle}{$0\vphantom{x_2}$\hh}
              \chunk[0]{$0$\hh}
              \chunk[1]{$0$\hh}
             \end{bundle}
          }
      \end{bundle}
   }
    \chunk[1]{
        \begin{bundle}{\fbox{$x_3\vphantom{1}$}\hh}
            \chunk[0]{
                   \begin{bundle}{\fbox{$x_3\vphantom{1}$}\hh}
                     \chunk[0]{$0$\hh}
                     \chunk[1]{\fbox{$1$}\hh}
                   \end{bundle}
            }
            \chunk[1]{
                   \begin{bundle}{\fbox{$x_3\vphantom{1}$}\hh}
                     \chunk[0]{$0$\hh}
                     \chunk[1]{\fbox{$1$}\hh}
                   \end{bundle}
            }
        \end{bundle}
    }
\end{bundle}%
}
&
\mbox{%
\setlength{\GapDepth}{1cm}%
\setlength{\GapWidth}{4mm}%
\setlength{\EdgeLabelSep}{0.5cm}%
\drawwith{\dottedline[ ]{3}}%
\begin{bundle}{\hh}
  \chunk[\protect\kreis{x_2}\protect\vphantom{1}]{
    \begin{bundle}{\hh}
      \chunk[\protect\kreis{x_1}]{
          \begin{bundle}{\hh}
             \chunk[\protect\kreis{x_3}]{\hh}
          \end{bundle}
      }
    \end{bundle}
  }
\end{bundle}%
}
&
\mbox{\setlength{\GapDepth}{1cm}%
\setlength{\GapWidth}{3mm}%
\setlength{\EdgeLabelSep}{0.5cm}%
\begin{bundle}{\fbox{$(\bar x_1\lor x_2)\land (\bar x_2\lor x_3)$}\hh}
   \chunk[0]{
      \begin{bundle}{\fbox{$\bar x_1$} \hh}
         \chunk[0\quad]{
            \begin{bundle}{\fbox{$1\vphantom{x_2}$}\hh}
              \chunk[0]{\fbox{$1$}\hh}
              \chunk[1]{\fbox{$1$}\hh}
             \end{bundle}
          }
      \end{bundle}
   }
    \chunk[1]{\fbox{$x_3\vphantom{1}$}\hh}
\end{bundle}%
}\\
& (a) &  & (b)\rule[0pt]{0pt}{7mm}
\end{tabular}
}
\caption{\label{fig:CT}(a) Elimination tree and (b) a $3$-cut for $\Phi=(\bar x_1\lor x_2)\land (\bar x_2\lor x_3)$,
due to the elimination order $(x_2,x_1,x_3)$.
The sum of the leaves is $\#\Phi=4$. Satisfiable  nodes are boxed. Note that $\#\Phi$ can already be computed from the
nodes on level $1$.}
\end{figure}

There is a nice interpretation of $\#k$-SAT in terms of complete
binary trees of height $n$ (i.\,e., having levels $0,\ldots,n$)
that is sometimes used in the context of counting.
An \emph{elimination tree} for a $k$-CNF $\Phi$ can be defined
as follows.
Fix an \emph{elimination order} $(y_1,\ldots,y_n)$ of the variables.
Every node $\phi$ of the tree corresponds to a Boolean formula.
The root (on level~$0$) of the tree is $\Phi$.
Every node $\phi$ on level $i$, $0\le i< n$, has two children:
One child is $\phi_{y_i=0}$, 
the other one is $\phi_{y_i=1}$. 
So a path from the root to a leaf corresponds to
an assignment to the variables, and the
formula at a leaf is either $0$ or $1$.
$\#\Phi$ is the number of leaves marked~$1$.
The mark $1$ is additionally broadcast to all internal
nodes on a path from a $1$-leaf to the root.
I.\,e., it is visible on every node $\phi$ whether $\phi$
is satisfiable or not.
For a small example, see Fig.~\ref{fig:CT}(a).

Let $\ell$ be a positive integer.
An $\ell$-cut of an elimination tree is
an arbitrary connected subtree that contains the root,
only $1$-nodes, and has $\ell$ leaves (w.\,r.\,t. the subtree).
For an example, see Fig.~\ref{fig:CT}(b).
An $\ell$-cut contains at most $n\cdot\ell$ nodes.
From determining an $\ell$-cut, immediately $\#\Phi\ge\ell$ follows.
Note that the elimination order significantly influences the moment
when in the elimination tree the number of satisfying assignments
can be determined.

Let $\varepsilon>0$ be an arbitrarily small number.
$\varepsilon$ is the upper bound on the admissible relative error.
A number $L$ is called an $\varepsilon$-approximation of $\#\Phi$
if $(1-\varepsilon)\cdot\#\Phi\le L\le(1+\varepsilon)\cdot\#\Phi$.
A \emph{randomized approximation scheme} (RAS)
$A$ is a randomized algorithm that computes
on inputs $\Phi$ and error $\varepsilon$ a number $A(\Phi,\varepsilon)$
such that
$
\Pr[ A(\Phi,\varepsilon)\text{ is an $\varepsilon$-approximation}]\ge \frac34
$. 
Note that by the \emph{median of means} method,
the probability can be boosted to any number $1-\delta$, for
$\delta$ being arbitrarily close to $0$.
Algorithm $A$ which usually outputs a mean
is repeated $R=\Theta(\log \delta^{-1})$ times,
and the median of the $R$ values computed is returned.

\subsection{Monte Carlo Counting}
\label{subsec:MC}

A very simple, general Monte Carlo approach for counting
works as follows.
Let $\cal S$ be the set whose cardinality has to be computed,
and let ${\cal U}\supseteq {\cal S}$ be a superset of $\cal S$
such that $|{\cal U}|$ can be computed exactly and, preferably, fast.
Sample independently and uniformly at random
$T$ elements from $\cal U$, and let $L$ be the number of elements
from $\cal S$ among these $T$ samples.
Return the number $\frac L T\cdot |{\cal U}|$
as an approximation of $|\cal S|$.
Standard probability theory (e.\,g.,
see~\cite[p.\,311]{MR:95}) 
gives
that a sample size of
$
T=\Theta(\varepsilon^{-2}\cdot |{\cal U}|/|{\cal S}|)
$
suffices to ensure 
the RAS property.

For $\#k$-SAT, ${\cal U}=\{0,1\}^n$ may be chosen.
If somehow a lower bound $\ell$ on $\#\Phi$
is known,
the Monte Carlo approach immediately gives a RAS with running time
$T_{\rm MC}=O(\varepsilon^{-2}\cdot 2^n/\ell\cdot (n+km))=O^*(\varepsilon^{-2}\cdot 2^n/\ell)$.
If $\ell$ (or $\#\Phi$) is small, this is unfortunately an unsatisfactory upper bound.
We refer to this algorithm as MC$(\Phi,\ell,\varepsilon)$.
If the reliability is amplified to $1-\delta$ by the median of means method,
the running time
is $O^*(\log(\delta^{-1})\cdot\varepsilon^{-2}\cdot 2^n/\ell)$, and
we write MC$(\Phi,\ell,\varepsilon,\delta)$.

Note that for the similar problem $\#$DNF where a
Boolean formula $\Phi$ in disjunctive normal form is given,
a set $\cal U$ can be devised~\cite{KLM:89} with
$|{\cal U}| \le~m\cdot~\#\Phi$
yielding a RAS with polynomial running time
$O(\varepsilon^{-2}\cdot m\cdot (n+km))=O^*(\varepsilon^{-2})$.

\subsection{Thurley's RAS}
\label{subsec:Thur}

The running time of MC is decreasing in $\ell$.
Therefore, Thurley presented an algorithm
that, for input $\Phi$ and $\ell$, determines
whether there are at least $\ell$ satisfying
assignments.
This time the running time is increasing in $\ell$.
In a last step, $\ell$ is chosen such that it balances
the running times of MC and Thurley's approach.
In the following, we explain this method in detail because
it is also the starting point for our improvements.

Let $\beta_k$ denote the smallest known constant
such that there is a randomized algorithm solving
$k$-SAT in time $O^*(2^{\beta_k\cdot n})$.
Hence, $\beta_3\approx 0.3864$,
$\beta_4\approx 0.5548$,
and, for $k\ge5$, $\beta_k=1-\mu_k/(k-1)$,
where
$$\mu_k=\sum_{j=1}^{\infty}\frac 1{j\cdot(j+\frac1{k-1})}$$
is the constant involved in the PPSZ algorithm~\cite{PPSZ:05}.


Thurley's RAS works as follows:
It has as input a $k$-CNF $\Phi$ and a bound $\ell$.
It computes an $\ell$-cut of the elimination tree
(if it exists).
Whether a node is a $1$-node can be checked fast
with the known randomized
$k$-SAT decision algorithms mentioned above where
$\delta^{-1}=2^{\Theta(n^c)}$, for some constant
$c\ge 2$. We call this the $\ell$-cut phase.

If it cannot find an $\ell$-cut, it reports the number of
subtree leaves it actually has found in the cut as estimation on $\#\Phi$
(in fact, this number is even the correct value, w.\,h.\,p.).
If, on the other hand, $\ell$ subtree leaves have been determined,
this means that a lower bound of $\ell$
on $\#\Phi$ has been determined.
In this case, the Monte Carlo algorithm MC$(\Phi,\ell,\varepsilon)$
is executed to approximate $\#\Phi$ with error $\varepsilon$.

Since in the worst case $2^i$ nodes (formulas) are on level $i$ and $i$ variables
have already been fixed,
the running time of the $\ell$-cut phase is
$$
T_{\rm cut} =
O^*\left(\sum_{i=0}^{\log (n\cdot\ell)} 2^i\cdot (2^{\beta_k})^{n-i}\right)
=
O^*\left(
2^{\beta_k\cdot n}\cdot\sum_{i=0}^{\log(n\cdot\ell)} (2^{1-\beta_k})^i
\right)
=
O^*(2^{\beta_k\cdot n}\cdot \ell^{1-\beta_k})
\enspace.
$$
Hence, the overall running time is at most
$T_{\rm cut}+T_{\rm MC}
=O^*(2^{\beta_k\cdot n}\cdot \ell^{1-\beta_k}+\varepsilon^{-2}\cdot 2^n/\ell)$
which becomes
$O^*\big(\varepsilon^{-2}\cdot\big(2^{1/(2-\beta_k)}\big)^n\big)$
when $\ell = 2^{n\cdot (1-\beta_k)/(2-\beta_k)}$.
For $k=3$, this is 
$O(\varepsilon^{-2}\cdot 1.5366^n)$,
and, for $k=4$, it is $O(\varepsilon^{-2}\cdot 1.6155^n)$.


The whole algorithm in both phases
does not exploit the actual structure of $\Phi$
which opens the possibility for improvements.


\section{Pruning the Tree: Taking Single Clauses into Account}
\label{sec:FirstImpr}

\begin{figure}
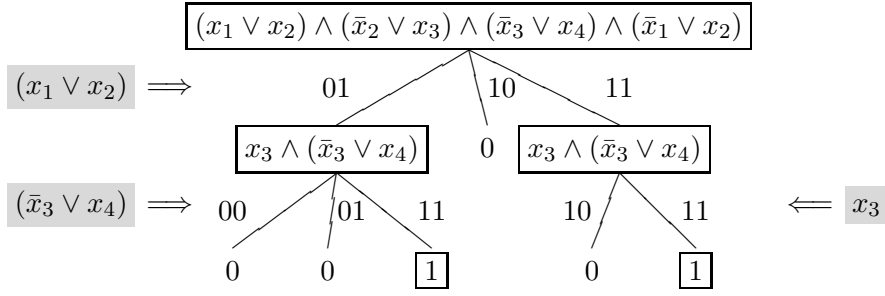
\centering
\scalebox{1}{
\begin{tabular}{ccc}
\qquad\qquad\qquad\mbox{%
\setlength{\GapDepth}{1cm}%
\setlength{\GapWidth}{4mm}%
\setlength{\EdgeLabelSep}{0.5cm}%
\drawwith{\dottedline[ ]{3}}%
\begin{bundle}{\hh}
  \chunk[\protect\kreis{(x_1\lor x_2)}\hspace*{2cm}\protect\vphantom{1}]{
    \begin{bundle}{\hh}
      \chunk[\protect\kreis{(\bar x_3\lor x_4)}\hspace*{2cm}]{\hh }
    \end{bundle}
  }
\end{bundle}%
}
&
\mbox{\setlength{\GapDepth}{1cm}%
\setlength{\GapWidth}{3mm}%
\setlength{\EdgeLabelSep}{0.5cm}%
\begin{bundle}{\fbox{$(x_1\lor x_2)\land (\bar x_2\lor x_3)\land (\bar x_3\lor x_4)\land (\bar x_1\lor x_2)$}\hh}
   \chunk[01]{
      \begin{bundle}{\fbox{$x_3\land (\bar x_3\lor x_4)$} \hh}
          \chunk[00]{$0$\hh}
          \chunk[\quad\,\, 01]{\qquad$0$\qquad\hh}
          \chunk[11]{\fbox{$1$}\hh}
      \end{bundle}
   }
   \chunk[\quad 10]{$0$\hh}
   \chunk[11]{
      \begin{bundle}{\fbox{$x_3\land (\bar x_3\lor x_4)$} \hh}
          \chunk[10\quad]{\qquad$0$\qquad\hh}
          \chunk[11]{\fbox{$1$}\hh}
      \end{bundle}
   }
\end{bundle}%
}
&
\mbox{%
\setlength{\GapDepth}{1cm}%
\setlength{\GapWidth}{4mm}%
\setlength{\EdgeLabelSep}{0.5cm}%
\drawwith{\dottedline[ ]{3}}%
\begin{bundle}{\hh}
  \chunk{
    \begin{bundle}{\hh}
      \chunk[\hspace*{2cm}\protect\rkreis{x_3}]{\hh }
    \end{bundle}
  }
\end{bundle}%
}
\end{tabular}
}
\caption{\label{fig:Cuts2vars}Possible pruned elimination tree
obtained by elimination w.\,r.\,t.
clause $(x_1\lor x_2)$ at the root and, therefore, chosen variables $x_1,x_2$,
and then on level 1, w.\,r.\,t. clause $(\bar x_3\lor x_4)$
on the left node, and clause $x_3$ on the right node.
For both nodes, the chosen variables are $x_3,x_4$.
Note that this tree has just $6$ leaves (the upper bound is $9$)
rather than $16$ leaves as in the
binary case from Sec.~\ref{sec:prelim}.
Also note that in general, on the same level different variables
at different nodes may be chosen.}
\end{figure}
The first possibility to slightly improve Thurley's RAS considers
the clauses that occur in the formulas $\phi$ that are associated
with the nodes of the elimination tree.
The following approach is also the basis for the further improvements
described in the subsequent sections.
If $C$ is a clause of~$\phi$ (in the following, we will write $C\in\phi$)
and consists of $\kappa$ ($\kappa\le k$)
literals, there is exactly one truth assignment to the $\kappa$ variables of $C$
such that $C$ is not satisfied.
So any assignment that incorporates this specific assignment
does not contribute to $\#\Phi$.

With this observation, it is possible to modify the
elimination tree as follows.
The root is the input $k$-CNF $\Phi$.
For a node's formula $\phi$,
choose one of its clauses, $C$, having $\kappa$ literals,
then choose $k$ variables including all variables from $C$,
and plug into the formula only those $2^k-2^{k-\kappa}$ ($<2^k$)
truth assignments that satisfy $C$.
For every formula obtained in this way, a new node is generated.
That means that we increase the degree of the elimination tree from $2$
to up to $2^k-1$ and reduce its height to $\lceil n/k\rceil$.
This tree is substantially smaller than the binary elimination
tree described in Sec.~\ref{sec:prelim}.
The total number of nodes is at most
$O((2^k-1)^{n/k})$. 
For $k=3$, this is $O(1.91294^n)$.
For a $2$-CNF example, see Fig.~\ref{fig:Cuts2vars}.

An $\ell$-cut of such an elimination tree
is now used in the cut phase of Thurley's RAS.
Now, the running time of the cut phase is
\begin{align}
T_{\rm Cut} &=
O^*\left(
\sum_{i=0}^{\log_{(2^k-1)}(n\cdot\ell)} (2^k-1)^i\cdot (2^{\beta_k})^{n-k\cdot i}
\right)\nonumber\\
&=
O^*(2^{\beta_k\cdot n}\cdot \ell^{1-\beta_k\cdot k/\log(2^k-1)})
\enspace.\label{eq:pruned}
\end{align}
If we use this pruned elimination tree to determine an $\ell$-cut, and then
(if necessary) using the Monte Carlo approach for computing the
estimate with error at most $\varepsilon$, we choose $\ell$ for
balancing the two phases as follows:
$$
\log\ell=\frac{1-\beta_k}{2-\beta_k\cdot\frac{k}{\log(2^k-1)}}\cdot n.
$$
We can state the result of this section with
$$
p_k = \frac{1-\beta_k\cdot\Big(\dfrac k {\log(2^k-1)}-1\Big)}
{2-\beta_k\cdot\dfrac k {\log(2^k-1)}},
$$

\begin{theorem}
Our pruned elimination tree algorithm is a RAS for $\#k$-SAT.
Its running time is $O(\varepsilon^{-2}\cdot (2^{p_k})^n)$.
\end{theorem}

So for $k=3$, the running time is
$O(\varepsilon^{-2}\cdot 1.5298^n)$, and
for $k=4$, it is
$O(\varepsilon^{-2}\cdot 1.6122^n)$
which is already slightly better than the running time of
Thurley's RAS (see Table~\ref{tab:KnownRes}).

From this first improvement approach we already learn that
by inspecting $\Phi$ we can reduce the size of the elimination tree.
But we can see more.
If we can simultaneously find several clauses that do not share variables,
we even can easily exclude from the elimination tree
all assignments that do not satisfy these clauses.

As the tree computed in this section may be unbalanced and different
variable choices on the same level are possible,
to sample assignments uniformly with the help of the tree
may be not possible.
That means that this pruned tree approach presumably
cannot be used to also speed up the Monte Carlo approach,
except for the transition from level 0 to level 1, where
the mentioned problem does not occur.

In the next sections, we construct larger structures that reduce
the size of the elimination tree further \emph{and} that
allow for faster uniform sampling.


\section{Taking Sets of Independent Clauses into Account}
\label{sec:NewRas}

For a $k$-CNF $\Phi$, two clauses of $\Phi$ are called independent
if they have no variable in common.
A subformula $\clset\subseteq\Phi$ 
is called independent if the
clauses in $\clset$ are pairwise independent. $\clset$ is called maximal
if every clause in $\Phi$ shares at least one variable with a clause in
$\clset$.
Let $|\clset|$ denote the number of clauses in $\clset$.


\subsection{Speeding up the Monte Carlo counting}
\label{subsec:MCspeed}

In the running time $T_{\rm MC}$ of the Monte Carlo algorithm
from Sec.~\ref{subsec:MC}, the cardinality of ${\cal U}$
plays a very important role.
Suppose that an independent subformula $\clset$ 
has been somehow (see Sec.~\ref{subsec:Red} below) computed.
Exploiting $\clset$, 
one can use a significantly smaller set ${\cal U}_{\clset}$, namely
${\cal U}_{\clset} = \{b \in \{0,1\}^n\bigm| b \text{ satisfies } {\clset}\}$
which is obviously a superset of
${\cal S}_{\Phi}=\{b \in \{0,1\}^n\bigm| b \text{ satisfies }\Phi\}$.
As the clauses in $\clset$ are independent,
the size  of ${\cal U}_{\clset}$ can be bounded as follows:
$|{\cal U}_{\psi}|\le
(2^k-1)^{|\clset|}\cdot 2^{n-k\cdot|\clset|}=2^n\cdot (1-2^{-k})^{|\clset|}$.
Sampling assignments from ${\cal U}_{\clset}$ is simply done by
choosing u.\,a.\,r. one of the
at most $2^k-1$ satisfying assignments to the variables 
of each $C\in{\clset}$, and assigning $0$ or $1$ to each of the remaining variables.
Provided $\ell\le\#\Phi=|{\cal S}_{\Phi}|$, the running time is
$
T_{\rm MC}
=O^*(
\frac{2^n}{\varepsilon^{2}\cdot\ell}
\cdot (1-2^{-k})^{|\clset|})
$.
In terms of elimination trees, $\clset$ is used during the sampling
for going from the root to the nodes on level $1$.
The remaining sampling is performed as in the binary elimination tree case
described in Sec.~\ref{subsec:MC}.


\subsection{Controlling the decision between elimination and recursion}
\label{subsec:Red}

Until now, we assumed $\clset$ already available.
Now we present a method for either obtaining a sufficiently large number
of independent clauses or, if this is not possible, to transform
the $k$-CNF 
to a not too large number of $(k-1)$-CNFs.
This method has been introduced by Hofmeister et al.~\cite{HSSW:02,HSSW:07}
and works on inputs $\Phi$ and integer $\hat m$ as follows.
In particular, it also controls whether the algorithm from
Sec.~\ref{sec:FirstImpr} is applied or recursive calls on
nodes of the elimination tree:

Starting with $\clset=\emptyset$, greedily increase $\clset$ as
much as possible.
Note that $\clset$ is now a maximal independent subformula.

If $|\clset|\ge \hat m$, return $\clset$.

Otherwise, use $\clset$ to generate the at most
$(2^k-1)^{\hat{m}}$ different $(k-1)$-CNFs on
level 1 of the elimination tree,
\emph{recursively} solve $\#(k-1)$-SAT with these formulas as input,
return the sum of the estimations,
and report the overall counting task as finished.


We will refer to this method as
\textsc{Red}$(\Phi,\hat{m},\varepsilon,\delta)$. 
When \textsc{Red}$(\Phi,\hat{m},\varepsilon,\delta)$ has been executed
and returned $\clset$ with  $\clset\ge \hat m$, store $\clset$
and if the pruned elimination tree
algorithm from Sec.~\ref{sec:FirstImpr} reports the existence of at least $\ell$
satisfying assignments, use $\clset$ in the Monte Carlo algorithm from
Sec.~\ref{subsec:MCspeed}. Since
\textsc{Red}
makes random decisions when solving the $\#(k-1)$-SAT instances,
we specify in the call of \textsc{Red}
also $\delta$ such that
the probability that \textsc{Red}
returns either an $\varepsilon$-approximation or a set $\clset$ of at
least $\hat{m}$ independent clauses is at least $1-\delta$.

Together with the modified algorithm to determine an $\ell$-cut from
the previous section, we obtain an improved RAS.
Since the algorithm for solving $\#k$-SAT for $k\ge 4$ includes
the algorithm itself for solving the occurring $\#(k-1)$-SAT instances,
the runtime has no closed form, but can be calculated recursively.
We terminate the recursion for $k=2$ and use
Wahlstr\"om's (even deterministic) algorithm~\cite{W:08}, solving $\#2$-SAT in time
$O(1.2377^n)$.
%
Although our method works for any $k$, we
state our result for the cases $k\in\{3,4\}$ 
only.

{\addtolength\leftmargini{1em} 

\begin{theorem}
The algorithm described above is a RAS solving $\#k${\rm-SAT}.
For $k=3$, its running time is $O(\varepsilon^{-2}\cdot 1.5181^n)$,
and for $k=4$, its running time is $O(\varepsilon^{-2}\cdot(\log(\delta^{-1})+n)\cdot 1.6105^n)$.
\end{theorem}

} 

\begin{proof}
(a) Let $k=3$. For arbitrary $\hat{m}$, \textsc{Red} either
finds a set $\clset$ of $\hat{m}$ independent clauses or
it has to solve at most $7^{\hat{m}}$ $\#2$-SAT instances,
each over the remaining $n-3\hat{m}$ variables.
This is done (even deterministically) with Wahlstr\"om's algorithm~\cite{W:08}
in time $O(1.2377^{n-3\hat{m}})$.
So the overall running time of \textsc{Red}
is
$T_{\textsc{Red}}=O( (7\cdot1.2377^{-3})^{\hat m}\cdot1.2377^n)
=O(3.6920^{\hat{m}}\cdot 1.2377^n)$. 

As mentioned above, the running time of MC in case of
$|{\clset}|\ge \hat{m}$ is at most
$T_{\rm MC}=
O^*(
2^n/\ell \cdot \left(\frac 7 8)^{\hat{m}}
\right)$
for $\ell\le\#\Phi$, and the running time for finding an $\ell$-cut (or
being sure that none exists) is due to Eq.~\eqref{eq:pruned},
$T_{\rm cut} = O^*(2^{\beta_3\cdot n}\cdot \ell^{1-\beta_3\cdot 3/\log 7})$.
The break-even point for $T_{\textsc{Red}}$,
$T_{\rm Cut}$ and $T_{\rm MC}$
and therefore the worst case occurs for choosing
$\hat{m} = 0.1563\,n$ and 
$\ell = 1.2903^n$.

\medskip

(b) The proof for the case $k\ge 4$ is analogous.
The only difference is that \textsc{Red} is not deterministic
anymore.
In order to get an $\varepsilon$-approximation with probability
at least $1-\delta$ in the end, we require from the nested
$\#(k-1)$-SAT algorithm an $\varepsilon$-approximation with
probability at least $1-\delta/7^{\hat{m}}\le 1-\delta/2^n$,
so the necessary number of repetitions for each recursive call of the
RAS for
$\#(k-1)$-SAT has to be increased by $n$ (additively)
to ensure the desired success probability.

For the time bound for $\#4$-SAT, we now use
the bound from (a) for $\#3$-SAT and obtain
the claimed result by the same straight-forward calculations
as in the first case. In the worst case, the parameters are
$\hat{m} = 0.0587\,n$ and 
$\ell = 1.2372^n$.
\end{proof}


\section{Taking Large Independent Structures into Account}
\label{sec:FinalRAS}

\newcommand{\final}{closed}

One can easily see that both the modified version of the Monte Carlo
algorithm and the modified method for calculating an $\ell$-cut do not
require the elements of $\clset$  to be clauses.
Both  can be generalized for $\clset$ being a set of pairwise 
independent, arbitrary subformulas with constant size.
E.\,g., in $\psi=\sigma_1\land\sigma_2$ with
$\sigma_1=(x_1\lor x_2\lor \bar x_3)\land (\bar x_1\lor x_3\lor x_4)$
and $\sigma_2=(x_5\lor \bar x_6)\land (x_5\lor x_7)$,
$\sigma_1$ and $\sigma_2$ are independent.
In order to recognize the subformulas, we write (in general)
$\psi=\{\sigma_1,\ldots,\sigma_{\hat{m}}\}$.
We call a single subformula $\sigma$ a \struct{}. If every clause shares at least one variable with a clause in $\clset$, we call $\clset$ maximal.
Due to their constant size, it is possible to compute the
number of their satisfying assignments.
We first show how \struct{}s can be used to improve
the algorithms MC and \textsc{Cut}, then we show how to construct them,
and finally we present the overall RAS.

\subsection{Monte Carlo counting and cut phase if many independent subformulas are known}
\label{subsec:DasRAS}

The generalized version of MC is presented below as Algorithm~\ref{alg:MC},
the generalized version of the $\ell$-cut phase
as \textsc{Cut} (Algorithm \ref{alg:Cut}).
For a subformula $\sigma$, let $n_{\sigma}$ denote the number
of different variables in $\sigma$, and $L_{\sigma}$
the number of satisfying assignments of $\sigma$.
Note that \textsc{Cut} 
uses a global counter $L$ (only set to $0$ in the first call) and
globally aborts as soon as $L$ reaches $\ell$. 
Otherwise, it returns $L$ after finishing.

Our results from Sec.~\ref{sec:NewRas}
lead to the following lemmas.

\begin{lem}\label{MC}
Let $\clset=\{\sigma_1,\ldots,\sigma_{\hat{m}}\}$ be a set of
constant-sized, pairwise
independent subformulas of $\Phi$.
If $\#\Phi\ge\ell$, MC($\Phi$, $\clset$, $\ell$, $\varepsilon$, $\delta$)
returns with probability at least $1-\delta$
an $\varepsilon$-approximation for $\#\Phi$.
The running time is 
$$
T_{\rm MC}=
O^*\!\Big(\log(\delta^{-1})\cdot
\frac{2^n}{\varepsilon^{2}\cdot\ell}\cdot \prod_{\sigma\in{\clset}}\frac{L_{\sigma}}{2^{n_{\sigma}}}
\Big)\enspace.
$$
\end{lem}

\begin{proof}
We sample from the set of all assignments satisfying $\clset$.
The size of this set is 
$$
2^{n-\sum_{{\sigma}\in{\clset}} n_{\sigma}}\cdot\prod_{{\sigma}\in{\clset}} L_{\sigma}
= 2^n\cdot \prod_{{\sigma}\in{\clset}}\frac{L_{\sigma}}{2^{n_{\sigma}}}
\enspace.
$$
The result follows directly from Sec.~\ref{subsec:MC}.
\end{proof}

Lemma \ref{MC} generalizes the result of Sec.~\ref{subsec:MCspeed},
where we used independent clauses, to the new situation, enabling the
use of more complex, but still pairwise independent subformulas.
The same generalization can be applied to the cut phase from 
Sec.~\ref{sec:FirstImpr}.

\begin{algorithm}[t]

\caption{MC$(\Phi,\clset,\ell,\varepsilon,\delta)$}
\label{alg:MC}

\SetKwInOut{Input}{input}
\SetKwInOut{Output}{output}

\Input{$k$-SAT instance $\Phi$,
set $\clset=\{\sigma_1,\ldots,\sigma_{\hat{m}}\}$ of independent subformulas,
parameters $\ell$, $\varepsilon$, $\delta$}
\Output{w.\,p. $1-\delta$: $\varepsilon$-approximation on $\#\Phi$}

\BlankLine

$\displaystyle U:=2^n\cdot\prod_{\sigma\in{\clset}} \frac{L_{\sigma}}{2^{n_{\sigma}}}$;\quad
$\displaystyle T:=\log(\delta^{-1})\cdot\frac{U}{\varepsilon^2\cdot\ell}$;\quad
$L:=0$\;
\SetKwFor{myRepeat}{repeat}{}{end}
\SetArgSty{}
\myRepeat{$T$ times}{
\SetArgSty{emph}
  \ForAll{$\sigma \in \clset$}{
    assign u.\,a.\,r. one of the assignments satisfying $\sigma$ to its variables\;
  }
  \ForAll{$x \in \operatorname{Vars}(\Phi)$ not already fixed}{
    assign u.\,a.\,r. $0$ or $1$ to $x$\;
  }
  \If{$\Phi$ is satisfied}{
    $L := L+1$\;
  }
}
\Return $L/T\cdot U$\;
\end{algorithm}

\begin{algorithm}
\caption{\textsc{Cut}$(\Phi,\clset,\ell,\delta)$ // Note: $L$ is a global variable}
\label{alg:Cut}

\SetKwInOut{Input}{input}
\SetKwInOut{Output}{output}

\Input{$k$-SAT instance $\Phi$, set of subformulas $\clset$, parameters $\ell$, $\delta$}
\Output{w.\,p. $1-\delta$: either $\#\Phi$ exactly, or the message that $\#\Phi\ge\ell$}

\BlankLine

\eIf{$\Phi$ is unsatisfiable (w.\,p. at least $1-\delta/2^n$)}{
  \Return $L$\;
}{
  $L := L+1$\;
  \If{$L\ge \ell$}{
    global abort\;
  }
}
\eIf {${\clset}\ne\emptyset$}{
  choose the first $\sigma\in\clset$\;
}{
  choose $\sigma \in \Phi$\;
}

\ForAll{Assignment $b$ to $\operatorname{Vars}(\sigma)$ satisfying $\sigma$}{
  \textsc{Cut}($\Phi_b$, ${\clset} \setminus\{\sigma\}$, $\ell$, $\delta$)\;
}
\Return $L$\;
\end{algorithm}

An important difference is that in Sec.~\ref{sec:FirstImpr},
we could assume that for each node $\phi$ there always is a clause
to generate new nodes since otherwise calculating $\#\phi$ would be trivial.
Now it may happen that
$\clset$ becomes empty. In this case, we fall 
back to using single
clauses, but due to the maximality of $\clset$,
those clauses have length at most $k-1$.
The following lemma states the running time of \textsc{Cut} 
in both cases.

\begin{lem}\label{Cut}
If $\#\Phi\le\ell$,
then, with probability at least $1-\delta$,
\textsc{Cut}($\Phi, \clset, \ell, \delta$) returns $\#\Phi$ exactly.
If $\ell\le\prod_{\sigma \in {\clset}} L_{\sigma}$, let ${\clset}'=\{\sigma_1,\ldots,\sigma_{\hat{m}'}\}\subset{\clset}$ denote
a subset of $\clset$, where $\hat{m}'$ is chosen such that $\prod_{i=1}^{\hat{m}'-1} L_{\sigma_i}<\ell\le\prod_{i=1}^{\hat{m}'} L_{\sigma_i}$.
If $\ell \ge \prod_{{\sigma} \in {\clset}} L_{\sigma}$, set $\hat{m}' := \hat{m}+\log_{(2^{k-1}-1)}(\ell \cdot \prod_{{\sigma} \in {\clset}} L_{\sigma}^{-1})$.\footnote{Roughly speaking, $\hat{m}'$ is the minimum number of elements of
$\clset$ respectively additional clauses
that have to be chosen
by \textsc{Cut} until the total of $\ell$ leaves of the elimination
tree can be achieved.}
Then the running time of \textsc{Cut} is
\begin{align*}
T_{\rm Cut} &=
\smash{O^*\left(\log(\delta^{-1})\cdot 2^{\beta_k\cdot n}\cdot
   \prod_{{\sigma} \in {\clset}'} \frac{L_{\sigma}}{2^{\beta_k\cdot n_{\sigma}}}\right)},\\
\intertext{if $\ell\le\prod_{\sigma \in {\clset}} L_{\sigma}$, and}
T_{\rm Cut} &= \smash{O^*\left(\log(\delta^{-1})\cdot 2^{\beta_k\cdot n}\cdot\left(\frac{2^{k-1}-1}{2^{\beta_k\cdot (k-1)}}\right)^{\hat{m}'-\hat{m}}\cdot\displaystyle \prod_{{\sigma} \in {\clset}} \frac{L_{\sigma}}{2^{\beta_k\cdot n_{\sigma}}}\right)}
\end{align*}
otherwise.
\end{lem}

\begin{proof}
Assume $\#\Phi\le\ell$, so the algorithm searches for a pruned
elimination tree with $\ell$ leaves.
Since $2^n$ is a rigorous bound on the number of formulas
that have to be tested on satisfiability,
the probability that the satisfiability check always gives
the right answer is at least $1-2^n\cdot\delta/2^n = 1-\delta$.

Now we focus on the run-time analysis. Basically, for each node of the tree,
at most $2^k-1$ different $k$-SAT instances must be solved.
So we have to analyze for each node $\phi$, how many variables of $\phi$
are fixed. 
On the $i$-th level of the tree, for $i\le\hat{m}$, the variables
of the first $i$ subformulas
of $\clset$ are fixed. If $i>\hat{m}$, all the variables of $\clset$ and additionally
the variables of $i-\hat{m}$ of the remaining clauses are fixed.
Note that due to the maximality of $\clset$, those clauses have size
at most $k-1$.
Let $N_i$ denote the number of nodes in level $i$. For $i\le \hat{m}$,
the nodes on level $i$ have $\sum_{j=1}^i n_{\sigma_j}$ fixed variables,
so the time $T_i$ required for processing only the nodes on level $i$ is
$$
T_i = O^*\left((n+\log(\delta^{-1}))\cdot N_i\cdot 2^{\beta_k\cdot n-\sum_{j=1}^i n_{\sigma_j}}\right)
 = O^*\left(\log(\delta^{-1})\cdot N_i\cdot 2^{\beta_k\cdot n}\cdot\prod_{j=1}^i 2^{-n_{\sigma_j}}\right).
$$
For $i>\hat{m}$, 
the number of fixed variables is $\sum_{j=1}^{\hat{m}} n_{\sigma_j}+(k-1)\cdot (i-\hat{m})$.
In this case, we obtain 
$$
T_i = O^*\left(\log(\delta^{-1})\cdot N_i\cdot 2^{\beta_k\cdot n}\cdot2^{-(k-1)\cdot(i-\hat{m})}\cdot\prod_{j=1}^{\hat{m}} 2^{-n_{\sigma_j}}\right).
$$
Since a node on level $i-1$ has out-degree at most $L_{\sigma_{i}}$ for $i\le \hat{m}$ and $2^{k-1}-1$ otherwise, on level $i$ there are at most $N_i\le\prod_{j=1}^i L_{\sigma_j}$ nodes if $i\le\hat{m}$ and $N_i\le\left(2^{k-1}-1\right)^{i-\hat{m}}\prod_{j=1}^{\hat{m}} L_{\sigma_j}$ nodes if $i\le\hat{m}$. Of course, $N_i\le \ell$ is also a bound we have for every $N_i$. Let $h$ be the maximum length of a path from the root to a leaf.
Then the overall running time $T_{\rm Cut}$ of \textsc{Cut} is (where we omit the $O^*(.)$ and the factor $\log(\delta^{-1})$)
\begin{align*}
T_{\rm Cut} &= \sum_{i=0}^h T_i = \sum_{i=0}^{\hat{m}} T_i + \sum_{i=\hat{m}+1}^h T_i\\
&\le \sum_{i=0}^{\hat{m}}\min\left\lbrace\ell,\prod_{j=1}^i L_{\sigma_j}\right\rbrace\cdot 2^{\beta_k\cdot n}\cdot\prod_{j=1}^i 2^{-\beta_k\cdot n_{\sigma_j}}\\
&\quad + \sum_{i=\hat{m}+1}^h \min\left\lbrace\ell,\left(2^{k-1}-1\right)^{i-\hat{m}}\prod_{j=1}^{\hat{m}} L_{\sigma_j}\right\rbrace\cdot 2^{\beta_k\cdot n}\cdot 2^{-(k-1)\cdot(i-\hat{m})}\cdot\prod_{j=1}^{\hat{m}} 2^{-n_{\sigma_j}}\\
&= 2^{\beta_k\cdot n}\cdot\sum_{i=0}^{\hat{m}}\min\left\lbrace\ell\cdot\prod_{j=1}^i 2^{-\beta_k\cdot n_{\sigma_j}},\prod_{j=1}^i \frac{L_{\sigma_j}}{2^{\beta_k\cdot n_{\sigma_j}}}\right\rbrace\\
&\quad + 2^{\beta_k\cdot n}\cdot\sum_{i=\hat{m}+1}^h \min\left\lbrace\ell\cdot 2^{-(k-1)\cdot(i-\hat{m})}\cdot\prod_{j=1}^{\hat{m}} 2^{-n_{\sigma_j}},\left(\frac{2^{k-1}-1}{2^{k-1}}\right)^{i-\hat{m}}\prod_{j=1}^{\hat{m}} \frac{L_{\sigma_j}}{2^{n_{\sigma_j}}}\right\rbrace\\
\end{align*}
One can easily see that in both sums, the first term inside the $\min$ is 
decreasing in $i$. Since for any $\sigma$ we are going to use, we have 
$L_{\sigma}/2^{\beta_k\cdot n_{\sigma}}>1$ (and since otherwise the 
runtime of the algorithm would be even better), the second term inside each 
$\min$ is increasing in $i$. So, up to a constant factor, the sum is equal to 
the summand where both terms inside the minimum have the same value, which
is the case for $i=\hat{m}'$. This finishes the proof.
\end{proof}

Note that, in the case of a node on level $i>\hat{m}$, \textsc{Cut}
would not have to solve a $k$-SAT instance anymore but only a
$(k-1)$-SAT instance.
However, such considerations would lead to
only a small improvement of the running time, so we
sacrificed the benefit from this observation in order to
simplify our analysis.

\subsection{Finding large sets of independent subformulas}
\label{sec:struct}

It remains to provide some method that collects the set $\clset$ of subformulas.
We call these subformulas \struct{}s and, in order to compute them,
we generalize algorithm \textsc{Red} from Sec.~\ref{subsec:Red}
as follows.
We start with the $\hat{m}$ independent clauses as initial set of \struct{}s.
Now we search iteratively for further clauses in $\Phi$ that extend in a controlled way
the \struct{}s we already have such that they remain independent.
Since the \struct{}s must have constant size, there are \struct{}s that we
do not want to extend anymore.
We call a variable $x$ occurring in \struct{} $\sigma$ \final{}, if we do not want to
add further clauses to $\sigma$ that contain $x$.
$\sigma$~is called \final{} if it has 
only \final{} variables.
For example, assume we already found the two \struct{}s
$\sigma_1=\{(x_1\lor x_2), (\bar x_2\lor x_3)\}$ and $\sigma_2=\{(x_4\lor x_5)\}$ where $x_2$ and $x_5$ are \final{} variables. 
The clause $(x_2\lor x_6)\in\Phi$ would not be considered because it contains the \final{} variable $x_2$.
But the clause $(x_3\lor x_4)\in\Phi$ can be used to extend and connect the two \struct{}s to
the single \struct{}
$\sigma_3=\{(x_1\lor x_2),
(\bar x_2\lor x_3),(x_3\lor x_4),(x_4\lor x_5)\}$ which is, in our example, a \final{} \struct{}.

Of course, there is no guarantee that the extension phase runs until every \struct{} is \final{}.
In the above example,
without the clause $(x_3\lor x_4)$ 
only the two non-\final{} \struct{}s can be obtained.
But if we ensure that each \struct{} 
has at least one \final{} variable in every clause, the set $\clset$ is maximal, meaning that every clause in $\Phi$ contains at least one variable that is a \final{} variable of some \struct{} $\sigma \in\clset$.
So, in order to reduce $\Phi$ to a $\#(k-1)$-SAT instance, one only
has to fix the \final{} variables. If \textsc{Red} finds a large number of \struct{}s, then
MC (Algorithm \ref{alg:MC})
runs faster. If there are only a few, then fixing their \final{} variables leads to only a few $\#(k-1)$-SAT
instances to be solved.
For a given \struct{} $\sigma$, let $f_\sigma$ denote the number of \final{} variables and $w_\sigma$
the number of assignments to the \final{} variables that do not already violate a clause of $\sigma$, 
i.\,e., $f_\sigma=n_\sigma$ and $w_\sigma=L_\sigma$ for every \final{} \struct{}.

The algorithm for grouping the \struct{}s is given in
Algorithm~\ref{alg:Red}. $\alpha_{k}$ refers to
a constant determined in the proof of Theorem~\ref{thm:DREI} below.
One can think of it as the value such 
that for the $\#k$-SAT problem an $\varepsilon$-approximation can be obtained in
time $O^*(\varepsilon^{-2}\cdot \log(\delta^{-1})\cdot \alpha_{k}^n)$
with probability at least $1-\delta$.
The algorithm accesses a library, containing for every
occurring \struct{} $\sigma$ the \final{} variables.
This library assures that the size of the \struct{}s is constant.

\begin{algorithm}[htb]

\caption{\textsc{Red}$(\Phi,\ell,\varepsilon,\delta)$}
\label{alg:Red}

\SetKwInOut{Input}{input}
\SetKwInOut{Output}{output}

\Input{$k$-SAT instance $\Phi$, parameters $\ell$, $\varepsilon$, $\delta$}
\Output{either a set $\clset$ of \struct{}s or, w.\,p. $1-\delta$: 
$\varepsilon$-approximation $L$ of $\#\Phi$}

\BlankLine

${\clset}:=\emptyset$\;
\While{there is a clause $C\in\Phi$ that contains no \final{} variable of a \struct{} in $\clset$}{
  $\chi:=$ all \struct{}s of $\clset$ that have a variable with $C$ in common\;
  $\sigma:=$ new \struct{} created from $C$ and the \struct{}s from $\chi$\; 
  ${\clset}:=({\clset}-\chi)\cup\{\sigma\}$\;
}\tcc{$\alpha_k$: const. determined in Thm.~\ref{thm:DREI}}
\eIf{$\alpha_{k-1}^n\cdot\prod_{\sigma\in{\clset}}w_\sigma\cdot \alpha_{k-1}^{-f_\sigma}< \alpha_k^n$}{
  $L:=0$\;
  \ForAll{assignments $b$ to the variables in $\clset$ that satisfy $\clset$}{
    $L:=L+\textsc{IndepSubform\_RAS}(\Phi_b,\varepsilon,\delta/2^n)$\;
  }
  \Return $L$\;
}{
  \Return $\clset$\;
}
\end{algorithm}

\begin{lem}\label{Red}
Assuming $\#(k-1)$-SAT can be approximated in time $O^*(\alpha_{k-1}^n)$, \textsc{Red}$(\Phi,\ell,\varepsilon,\delta)$ (Algorithm \ref{alg:Red})
returns either w.\,p. at least $1-\delta$ an $\varepsilon$-approxima\-tion of $\#\Phi$
or a set $\clset$ of pairwise independent \struct{}s
 such that $\alpha_{k-1}^n\cdot\prod_{\sigma\in{\clset}}w_\sigma/\alpha_{k-1}^{f_\sigma}\ge \alpha_k^n$.
It runs in time
$$
O^*\left(\varepsilon^{-2}\cdot\log(\delta^{-1})\cdot \alpha_{k-1}^n\cdot\prod_{\sigma\in{\clset}}\frac{w_\sigma}{\alpha_{k-1}^{f_\sigma}}\right).
$$
\end{lem}
\begin{proof}
If 
$\alpha_{k-1}^n\cdot\prod_{\sigma\in{\clset}}\frac{w_\sigma}{\alpha_{k-1}^{f_\sigma}}>t$
after the \textbf{while}-loop, which obviously runs in polynomial time, is finished,
then the algorithm returns
the set $\clset$ with the claimed property. Otherwise it enumerates 
all assignments for the \final{} variables of the \struct{}s of $\clset$ that do not already cause $\clset$ to be evaluated to $0$. Since the \struct{}s in $\clset$ are pairwise independent, there are exactly $\prod_{\sigma\in{\clset}}w_\sigma$ such assignments.
For each of these assignments, \textsc{Red} starts an
algorithm that w.\,p. at least $1-\delta/2^n$ returns an $\varepsilon$-approximation of the
resulting $\#(k-1)$-SAT instances. The probability that each of the $\prod_{\sigma\in{\clset}}w_\sigma$ runs actually returns an $\varepsilon$-approximation is therefore at least
$1-\delta/2^n\cdot\prod_{\sigma\in{\clset}}w_\sigma\ge 1-\delta$. Since $\sum_{\sigma\in\clset} f_\sigma$ variables are fixed, each of the $\#(k-1)$-SAT instances has $n-\sum_{\sigma\in\clset} f_\sigma$ variables and to approximate it within the given parameters takes time 
$$
O^*\left(\varepsilon^{-2}\cdot(n+\log(\delta^{-1}))\cdot \alpha_{k-1}^{n-\sum_{\sigma\in\clset} f_\sigma}\right) =
O^*\left(\varepsilon^{-2}\cdot\log(\delta^{-1})\cdot \alpha_{k-1}^n \cdot \prod_{\sigma\in\clset} \alpha_{k-1}^{-f_\sigma}\right).
$$
This finishes the proof.
\end{proof}

\subsection{All things come together: The new randomized approximation scheme}
\label{subsec:HierIsses}

Combining all results, 
we are now able to state our main algorithm that solves
$\#k$-SAT.
\textsc{IndepSubform\_RAS}
(Algorithm \ref{alg:RAS}) shows how to combine the Algorithms \ref{alg:MC}, \ref{alg:Cut} and \ref{alg:Red}.

\begin{algorithm}[htb]

\caption{\textsc{IndepSubform\_RAS}}
\label{alg:RAS}

\SetKwInOut{Input}{input}
\SetKwInOut{Output}{output}

\Input{$k$-CNF $\Phi$, $\varepsilon$, $\delta$}
\Output{Approximated number $L$ of satisfying assignments of $\Phi$}

\BlankLine

\eIf {\textsc{Red}$(\Phi$, $\ell$, $\varepsilon$, $\delta)$ returns $L$}{ 
  \Return $L$
}{
  $\clset:=$ set of \struct{}s returned by \textsc{Red}\;
  \eIf {\textsc{Cut}$(\Phi$, $\clset$, $\ell$, $\delta)$ returns value $\#\Phi<\ell$}{
    \Return $\#\Phi$\;
  }{
    \Return MC($\Phi$, $\clset$, $\ell$, $\varepsilon$, $\delta$)\;
  }
}
\end{algorithm}

\begin{theorem}\label{thm:DREI}
\textsc{IndepSubform\_RAS} (Algorithm \ref{alg:RAS}) is a RAS
running in time $O(\varepsilon^{-2}\cdot\log(\delta^{-1})\cdot\theConstdrei^n)$ for
$\#3$-SAT and in time $O(\varepsilon^{-2}\cdot\log(\delta^{-1})\cdot \theConstvier^n)$ for
$\#4$-SAT.
\end{theorem}
\begin{proof}
The results follow from calculating the break-even points of the time bounds in
Lemma \ref{MC}, Lemma \ref{Cut} and Lemma \ref{Red}. For our version of algorithm \textsc{Red}, we
used the \struct{}s and made the decisions about which variable to declare \final{}
as described in Table \ref{tab} and \ref{tab2}, resp. We declared every other \struct{} that
is not listed as \final{} (by setting all its variables \final{}) and set in Algorithm \ref{alg:Red} $\alpha_3=\theConstdrei$ for $k=3$ and 
$\alpha_4=\theConstvier$ for $k=4$. 

\begin{table}[htb]
\noindent\begin{minipage}[t]{0.49\textwidth}
\caption{Non-\final{} \struct{}s for $k=3$}
\medskip\centering
\begin{tabular}{|l|l|l|}
\hline
Type $\sigma$ & $L_\sigma$ & \final{} Vars.\\
\hline
\hline
$\{(x_1 \lor x_2 \lor x_3)\}$ & $7$ & $\{x_3\}$\\
\hline
$\{(x_1 \lor x_2 \lor x_3),$ & \multirow{2}{0.8cm}{$25$} & \multirow{2}{0.8cm}{$\{x_1\}$}\\
$(x_1\lor x_4\lor x_5)\}$ & &\\
\hline
$\{(x_1 \lor x_2 \lor x_3),$ & \multirow{2}{0.8cm}{$13$} & \multirow{2}{0.8cm}{$\{x_1\}$}\\
$(x_1\lor x_2\lor x_4)\}$ & &\\
\hline
$\{(x_1 \lor x_2 \lor x_3),$ & \multirow{3}{0.8cm}{$89$} & \multirow{3}{0.8cm}{$\{x_1,x_2\}$}\\
$(x_1\lor x_4\lor x_5),$ & &\\
$(x_2\lor x_6\lor x_7)\}$ & & \\
\hline
\end{tabular}
\label{tab}
\end{minipage}\hfill%
%
\begin{minipage}[t]{0.49\textwidth}
\caption{Non-\final{} \struct{}s for $k=4$}
\medskip\centering
\begin{tabular}{|l|l|l|}
\hline
Type $\sigma$ & $L_\sigma$ & \final{} Vars.\\
\hline
\hline
$\{(x_1 \lor x_2 \lor x_3 \lor x_4)\}$ & $15$ & $\{x_4\}$\\
\hline
$\{(x_1 \lor x_2 \lor x_3 \lor x_4),$ & \multirow{2}{0.8cm}{$113$} & \multirow{2}{0.8cm}{$\{x_1\}$}\\
$(x_1\lor x_5\lor x_6\lor x_7)\}$ & &\\
\hline
$\{(x_1 \lor x_2 \lor x_3 \lor x_4),$ & \multirow{3}{0.8cm}{$851$} & \multirow{3}{0.8cm}{$\{x_1,x_2\}$}\\
$(x_1\lor x_5\lor x_6\lor x_7),$ & &\\
$(x_2\lor x_8\lor x_9\lor x_{10})\}$ & & \\
\hline
\end{tabular}
\label{tab2}
\end{minipage}
\end{table}

For $k=3$, in the worst case, $\ell = 1.28794^n$ and there are $0.05252\,n$ \final{} \struct{}s of the form 
$
\{(x_1\lor x_2\lor x_3),(x_1\lor x_4\lor x_5),(x_2\lor x_6\lor x_7),(x_4\lor x_8\lor x_9)\}.
$
For $k=4$, in the worst case, $\ell = 1.23823^n$ and there are $0.01785\, n$ \final{} \struct{}s of the form 
$
\{(x_1\lor x_2\lor x_3\lor x_4),(x_1\lor x_5\lor x_6\lor x_7),(x_2\lor x_8\lor x_9\lor x_{10}),(x_5\lor x_{11}\lor x_{12}\lor x_{13})\}.
$
This leads to the claimed bounds.
\end{proof}

Note that, for any $k$, our method runs in time
$O^*(\alpha_k^n)$ with $\alpha_k$ depending on $\beta_k$
and $\alpha_{k-1}$.
For all $k$, $\alpha_k < 2^{1/(2-\beta_k)}=:\vartheta_k$
(Thurley's running time), even
if we define \struct{}s consisting of just a single clause as \final{}.
E.\,g., 
$\alpha_5 \approx 1.6694 < \vartheta_5\approx 1.6712$.



\bibliography{literature}

\end{document}